%% file: main.tex
\documentclass[letterpaper, 10 pt, conference]{ieeeconf}
\usepackage{cite}
\usepackage{amsmath,amssymb,amsfonts}
\usepackage{bm}
\usepackage{algorithmic}
\usepackage{graphicx}
\usepackage{textcomp}
\usepackage{xcolor,color}
\usepackage[ruled,vlined]{algorithm2e}
\usepackage{hyperref}
\usepackage{tikz}
\usetikzlibrary{positioning}
    
\makeatletter
\let\MYcaption\@makecaption
\makeatother 

\usepackage[font=footnotesize]{subcaption}

\makeatletter
\let\@makecaption\MYcaption
\makeatother

\newtheorem{definition}{Definition}

\newtheorem{theorem}{Theorem}

\newtheorem{remark}{Remark}


\IEEEoverridecommandlockouts
\begin{document}

\title{\bf Robust Approximate Simulation for Hierarchical Control of Piecewise \\ Affine Systems under Bounded Disturbances}

\author{Zihao Song, Vince Kurtz, Shirantha Welikala, Panos J. Antsaklis and Hai Lin\thanks{This work was supported by the National Science Foundation under Grant IIS1724070, Grant CNS-1830335, and Grant IIS-2007949. The authors are with the Department of Electrical Engineering, University of Notre Dame, Notre Dame, IN 46556 USA (e-mail: zsong2@nd.edu; vjkurtz@gmail.com; wwelikal@nd.edu; hlin1@nd.edu; pantsakl@nd.edu). The first author sincerely appreciates Dr. P. M. Wensing for several fruitful discussions.}}
% \vspace{-10mm}
\maketitle
\thispagestyle{empty}

\begin{abstract}
    Piecewise affine (PWA) systems are widely applied in many practical cases such as the control of nonlinear systems and hybrid dynamics. However, most of the existing PWA control methods have poor scalability with respect to the number of modes and system dimensions and may not be robust to the disturbances in performance. In this paper, we present a robust approximate simulation based control method for PWA systems under bounded external disturbances. First, a lower-dimensional linear system (abstraction) and an associated \emph{interface} are designed to enable the output of the PWA system (concrete system) to track the output of the abstraction. Then, a Lyapunov-like \emph{simulation function} is designed to show the boundedness of the output errors between the two systems. Furthermore, the results obtained for linear abstraction are extended to the case that a simpler PWA system is the abstraction. To illustrate the effectiveness of the proposed approach, simulation results are provided for two design examples.
\end{abstract}

% \vspace{-4mm}
\section{Introduction}\label{sec:intro}

Piecewise Affine (PWA) systems are an important and powerful modeling tool that can globally approximate nonlinear systems with finite number of linear characteristics on partitions \cite{Rodrigues2001ACDPWA}. In this case, controllers can be designed based on linear system control theories and better accommodate the existing nonlinearities. Because of this advantage, PWA systems are applicable in many engineering fields, and thus have attracted a plenty of attention since it was first proposed. Moreover, PWA systems are particularly used in situations where a plant operates in different modes or under physical constraints such as in robot locomotion \cite{Han2017FDMCPR}, manipulation with contacts \cite{Hogan2018RPM} and hybrid control systems \cite{Habets2006RCSPWAHS}. 

In recent years, together with Mixed Logical Dynamical (MLD) systems, PWA systems are the most popular modeling framework for hybrid systems in Model Predictive Control (MPC) \cite{CAMACHO201021MPCHS}. However, such on-line optimal control methods typically lead to the synthesis of very inefficient Mixed-Integer Convex Programs (MICPs) with large computational price \cite{Marcucci2019MIFOCPWA}.  To reduce the complexity of computing the convex polyhedra for PWA systems, an optimal sampling-based controller is proposed in \cite{sadraddini2018arXiv}, which can be cast into a single MICP problem. Nevertheless, for general PWA systems, the aforementioned on-line control strategies deteriorate with the number of partitions and the dimension of the systems. Similarly, the existing off-line controller design approaches for PWA (e.g., the LMI based controller synthesis approaches \cite{Hassibi1998QSCPWA}) scale poorly for high dimensional PWA systems. Therefore, the controller synthesis for high-dimensional PWA systems with a great number of partitions is still an open problem.

Approximate simulation is an extension of simulation relations from formal methods to continuous systems. It is a powerful tool often used for hierarchical control of complex and high-dimensional systems \cite{Girard2006HCAS,GIRARD2009HCSD,Vince2020RASHC}. This approximate simulation framework defines an approximate relationship between two systems, i.e. the full-order (concrete) system and the reduced-order system (abstraction). By constructing an abstraction and a corresponding \emph{interface} between this abstraction and the concrete system, the outputs of both systems can be guaranteed to remain close within some certain error bound characterized by a Lyapunov-like \emph{simulation function}. Different from the asymptotic model matching \cite{Di1994AMMNS} that enforces output global asymptotic stability of the involved systems, it is not required in approximate simulation framework that the trajectories of the system and its abstraction match exactly but only approximately. This relaxation allows us to consider simpler abstractions and thus simplify the design of high-level control tasks. 

Motivated by the above observations, we extend the results in \cite{Vince2020RASHC} to consider the case where the concrete system is a PWA system. Our primary contributions can be summarized as follows:
\begin{enumerate}
    \item A novel robust approximate simulation based control strategy is proposed for the control of PWA systems using a linear system as the abstraction;
    \item The results obtained in the linear system abstraction approach is generalized for the case that the abstraction is a simpler PWA system.
\end{enumerate}

The remainder of this paper is organized as follows. The problem formulation and some necessary preliminaries are presented in Section \ref{sec:background}. Our main results are presented in Section \ref{sec:main_results}, and are supported by simulation examples in Section \ref{sec:simulation}. Finally, concluding remarks are provided in Section \ref{sec:conclusion}.

%-------------------------------------------
% \vspace{-3mm}
\section{Background}\label{sec:background}

\subsection{Problem Formulation}

Consider a piecewise-affine (PWA) system as follows
\begin{equation}\label{eq:PWA_concrete}
    \Sigma:\left\{\begin{array}{ll}
    \mathbf{\dot{x}}_1=\mathbf{A}_i\mathbf{x}_1+\mathbf{B}_i\mathbf{u}_1+\mathbf{c}_i &  \\
    \mathbf{y}_1=\mathbf{C}_i\mathbf{x}_1 &    
    \end{array}\right.
\end{equation}
where the system state \(\mathbf{x}_1\in \mathcal{X}_1^{i}:=\{\mathbf{x}_1\in\mathbb{R}^n|\ \mathbf{E}_i\mathbf{x}_1\geq\mathbf{f}_i\}\), and \(\mathbf{E}_i\in\mathbb{R}^{b\times n}\), \(\mathbf{f}_i\in\mathbb{R}^b\), \(\mathbf{A}_i\in\mathbb{R}^{n\times n}\), \(\mathbf{B}_i\in\mathbb{R}^{n\times p}\) and \(\mathbf{C}_i\in\mathbb{R}^{k\times n}\) are given, for \(i\in\mathcal{I}:=\{1,...,s\}\). The system output is \(\mathbf{y}_1\in\mathbb{R}^k\). The term \(\mathbf{c}_i\in\mathbb{R}^{n}\), which represents the lumped bounded external disturbances and piecewise linearization error, is assumed to be bounded such that \(||\mathbf{c}_i||_{\infty}\leq \bar{c}_i\), where \(\bar{c}_i\) is known. The cell bounding \cite{johansson2003piecewise} is defined as \(\mathbf{\bar{E}}_i=\begin{bmatrix}
    \mathbf{E}_i & -\mathbf{f}_i
\end{bmatrix}\) for the partitions with some of the boundaries not crossing the origin (denoted as \(\mathcal{I}_1\)), with \(\mathbf{\bar{E}}_i\begin{bmatrix}
    \mathbf{x}_1^T & 1 
\end{bmatrix}^T\geq \mathbf{0}\) and it reduces to \(\mathbf{E}_i\) for the partitions with all their boundaries crossing the origin (denoted as \(\mathcal{I}_0\)), with \(\mathbf{E}_i\mathbf{x}_1 \geq \mathbf{0}\). Besides, the continuity matrix \cite{johansson2003piecewise} is defined as \(\mathbf{\bar{J}}_i=\begin{bmatrix}
    \mathbf{J}_i & h_i
\end{bmatrix}\) for \(i\in\mathcal{I}_1\) (and \(\mathbf{J}_i\) for \(i\in\mathcal{I}_0\)), with \(\mathbf{\bar{J}}_{i_1}\begin{bmatrix}
    \mathbf{x}_1^T & 1 
\end{bmatrix}^T=\mathbf{\bar{J}}_{i_2}\begin{bmatrix}
    \mathbf{x}_1^T & 1 
\end{bmatrix}^T\) for \(\mathcal{X}_1^{i_1}\cap\mathcal{X}_1^{i_2}\), \(i_1\), \(i_2\in \mathcal{I}\).

In this paper, we first consider the abstraction as a linear system of the following form
\begin{equation}\label{eq:linear_abstraction}
    \Sigma':\left \{\begin{array}{ll}
    \mathbf{\dot{x}}_2=\mathbf{F}\mathbf{x}_2+\mathbf{G}\mathbf{u}_2 &  \\
    \mathbf{y}_2=\mathbf{H}\mathbf{x}_2 &    
    \end{array}\right.
\end{equation}
where $\mathbf{x}_2 \in \mathbb{R}^m$ is the state of system (\ref{eq:linear_abstraction}), $\mathbf{y}_2 \in \mathbb{R}^k$ is the system output. The matrices \(\mathbf{F}\in\mathbb{R}^{m\times m}\), \(\mathbf{G}\in\mathbb{R}^{m\times q}\) and \(\mathbf{H}\in\mathbb{R}^{k\times m}\) are free to select. The abstraction (\ref{eq:linear_abstraction}) is typically simpler than each mode of the PWA system \(\Sigma\) in terms of system dimension, i.e., \(m\leq n\). 

Then, if a single linear abstraction may not be completely adequate (typically if the concrete PWA system \(\Sigma\) has many modes), the abstraction can be selected as a simpler PWA system of the form
\begin{equation}\label{eq:PWA_abstraction}
    \Sigma'':\left \{\begin{array}{ll}
    \mathbf{\dot{x}}_2=\mathbf{F}_j\mathbf{x}_2+\mathbf{G}_j\mathbf{u}_2 &  \\
    \mathbf{y}_2=\mathbf{H}_j\mathbf{x}_2 &    
    \end{array}\right.
\end{equation}
where \(\mathbf{x}_2\in \mathcal{X}_2^{j}:=\{\mathbf{x}_2\in\mathbb{R}^m|\ \mathbf{E}_{aj}\mathbf{x}_2\geq\mathbf{f}_{aj}\}\), \(\mathbf{E}_{aj}\in\mathbb{R}^{d\times m}\), \(\mathbf{f}_{aj}\in\mathbb{R}^d\), $\mathbf{y}_2 \in \mathbb{R}^k$, for \(j\in\mathcal{I}_a:=\{1,...,r\}\) with \(r\leq s\). The matrices \(\mathbf{F}_j\in\mathbb{R}^{m\times m}\), \(\mathbf{G}_j\in\mathbb{R}^{m\times q}\) and \(\mathbf{H}_j\in\mathbb{R}^{k\times m}\) are free to select, for \(j\in\mathcal{I}_a\). Let us denote the partitions with all their boundaries crossing the origin as \(\mathcal{I}_{a0}\) and the partitions with some of the boundaries not crossing the origin as \(\mathcal{I}_{a1}\). It is important to note that the abstraction \(\Sigma''\) has less modes and much simpler dynamics than the concrete system \(\Sigma\). The details of the partitions of the abstract PWA system will be studied in Section III.

Consider the abstraction in (\ref{eq:linear_abstraction}). Define \(\mathbf{\tilde{x}}=\mathbf{x}_1-\mathbf{P}_i\mathbf{x}_2\), where \(\mathbf{P}_i\in\mathbb{R}^{n\times m}\) is an injective map from the state space of \(\mathbf{x}_2\) to that of \(\mathbf{x}_1\) for \(i\in\mathcal{I}\). Then,
\begin{align*}
    \mathbf{\dot{\tilde{x}}}&=\mathbf{\dot{x}}_1-\mathbf{P}_i\mathbf{\dot{x}}_2\\
    &=(\mathbf{A}_i\mathbf{x}_1+\mathbf{B}_i\mathbf{u}_1+\mathbf{c}_i)-\mathbf{P}_i(\mathbf{F}\mathbf{x}_2+\mathbf{G}\mathbf{u}_2)\\
    &=\mathbf{A}_i\mathbf{\tilde{x}}+\mathbf{B}_i\mathbf{u}_v+\mathbf{A}_i\mathbf{P}_i\mathbf{x}_2+\mathbf{c}_i-\mathbf{P}_i(\mathbf{F}\mathbf{x}_2+\mathbf{G}\mathbf{u}_2)
\end{align*}
where \(\mathbf{u}_v:=\mathbf{u}_1(\mathbf{\tilde{x}},\mathbf{x}_2,\mathbf{\bar{u}}_2)\) is the control input of the concrete system, which is also known as the \emph{interface} that we need to design. For the abstraction (\ref{eq:linear_abstraction}), an input transformation law is designed as
\begin{equation}\label{eq:abstraction_ioft}
    \mathbf{u}_2=\mathbf{L}\mathbf{x}_2+\mathbf{\bar{u}}_2
\end{equation}
where the matrix \(\mathbf{L}\in\mathbb{R}^{q\times m}\) is selected such that all the eigenvalues of the matrix \(\mathbf{F}+\mathbf{G}\mathbf{L}\) of the transformed abstraction \(\Sigma'\) have negative real parts, and \(\mathbf{\bar{u}}_2\in\mathbb{R}^q\) is the transformed input. Thus, we define our \emph{robust approximate simulation framework} as the following joint system 
\begin{equation}\label{eq:entire_system}
    \left \{\begin{array}{ll}
    \mathbf{\dot{\tilde{x}}}=\mathbf{A}_i\mathbf{\tilde{x}}+\mathbf{B}_i\mathbf{u}_v-[\mathbf{B}_i\mathbf{Q}_i+\mathbf{P}_i\mathbf{G}\mathbf{L}]\mathbf{x}_2-\mathbf{P}_i\mathbf{G}\mathbf{\bar{u}}_2+\mathbf{c}_i &  \\
    \mathbf{\dot{x}}_2=(\mathbf{F}+\mathbf{G}\mathbf{L})\mathbf{x}_2+\mathbf{G}\mathbf{\bar{u}}_2 &\\
    \mathbf{e}=\mathbf{y}_1-\mathbf{y}_2=\mathbf{C}_i\mathbf{\tilde{x}} &    
    \end{array}\right.
\end{equation}
for \(i\in\mathcal{I}\), where \(\mathbf{P}_i\in\mathbb{R}^{n\times m}\) and \(\mathbf{Q}_i\in\mathbb{R}^{n\times m}\) satisfy the conditions:
\begin{equation}\label{approximate_relation_i}
    \mathbf{H}=\mathbf{C}_i\mathbf{P}_i,\ \mathbf{P}_i\mathbf{F}=\mathbf{A}_i\mathbf{P}_i+\mathbf{B}_i\mathbf{Q}_i
\end{equation}
Define the state of the joint system as \(\boldsymbol{\omega}=\begin{bmatrix}
    \mathbf{\tilde{x}}^T & \mathbf{x}_2^T \\
\end{bmatrix}^T\). Then, the joint partitions can be written as
    \begin{equation}\label{eq:entire_system_partition}
        \begin{split}
        \boldsymbol{\omega}\in \mathbf{\Omega}_i:&=\{\mathbf{E}_i(\mathbf{x}_1-\mathbf{P}_i\mathbf{x}_2)+\mathbf{E}_i\mathbf{P}_i\mathbf{x}_2\geq\mathbf{f}_i\}\\
        &=\{[\mathbf{E}_i\ \mathbf{E}_i\mathbf{P}_i]\boldsymbol{\omega}\geq\mathbf{f}_i\}
        \end{split}
    \end{equation}
Note that the joint partition (\ref{eq:entire_system_partition}) can be equivalently represented by:    
    \begin{equation}\label{eq:entire_system_partition_1}
        \begin{split}
        \boldsymbol{\bar{\omega}}\in \mathbf{\bar{\Omega}}_i:&=\{\begin{bmatrix}
        \mathbf{E}'_i & \mathbf{f}'_i
        \end{bmatrix}\begin{bmatrix}
        \boldsymbol{\omega} \\
        1
       \end{bmatrix}\geq \mathbf{0}\}=\{\mathbf{\bar{E}}_i\boldsymbol{\bar{\omega}}\geq \mathbf{0}\}
        \end{split}
    \end{equation}
where \(\mathbf{E}'_i=\begin{bmatrix}
    \mathbf{E}_i & \mathbf{E}_i\mathbf{P}_i
\end{bmatrix}\) and \(\mathbf{f}'_i=-\mathbf{f}_i\). 

The joint system (\ref{eq:entire_system}) corresponding to the PWA system abstraction (\ref{eq:PWA_abstraction}) is similar  to that of the linear system abstraction (\ref{eq:linear_abstraction}). In particular, it will result in a joint system of the form (\ref{eq:entire_system}) for each tuple (\(\mathbf{F}_j\), \(\mathbf{G}_j\), \(\mathbf{L}_j\)), for \(j\in\mathcal{I}_a\), and the conditions in (\ref{approximate_relation_i}) become
\begin{equation}\label{approximate_relation_ij}
    \mathbf{H}_j=\mathbf{C}_i\mathbf{P}_i,\ \mathbf{P}_i\mathbf{F}_j=\mathbf{A}_i\mathbf{P}_i+\mathbf{B}_i\mathbf{Q}_i
\end{equation}
for \(i\in\mathcal{I}\). In other words, for each \(i\in \mathcal{I}\) we need to determine matrices \(\mathbf{P}_i\) and \(\mathbf{Q}_i\) such that (\ref{approximate_relation_ij}) holds true for a pair (\(\mathbf{F}_j\), \(\mathbf{H}_j\)), \(j\in \mathcal{I}_a\). The details of the joint partitions corresponding to this case will be discussed in Section III.

\begin{remark}
    The input transformation law (\ref{eq:abstraction_ioft}) is introduced to increase the tunability of the matrix \(\mathbf{F}\) (or \(\mathbf{F}_j\) for the PWA abstraction (\ref{eq:PWA_abstraction})). Besides, it will be shown in Section III that some LMI based conditions should be satisfied with all the eigenvalues of (\(\mathbf{F}+\mathbf{G}\mathbf{L}\)) (or (\(\mathbf{F}_j+\mathbf{G}_j\mathbf{L}_j\)) for the PWA abstraction (\ref{eq:PWA_abstraction})) having negative real parts.
\end{remark}

The objective of this paper is to design an interface \(\mathbf{u}_v:=\mathbf{u}_1(\mathbf{\tilde{x}},\mathbf{x}_2,\mathbf{\bar{u}}_2)\) for the joint system (\ref{eq:entire_system}) over the joint partitions (e.g.,(\ref{eq:entire_system_partition})) to guarantee the boundedness of output error between \(\mathbf{y}_1\) and \(\mathbf{y}_2\), i.e., \(||\mathbf{e}||=||\mathbf{y}_1-\mathbf{y}_2||\leq \delta\), where \(\delta\in\mathbb{R}_{\geq 0}\) is some certain error bound. The basic control architecture proposed in this paper is shown in Figure \(\ref{fig:control_arch}\).

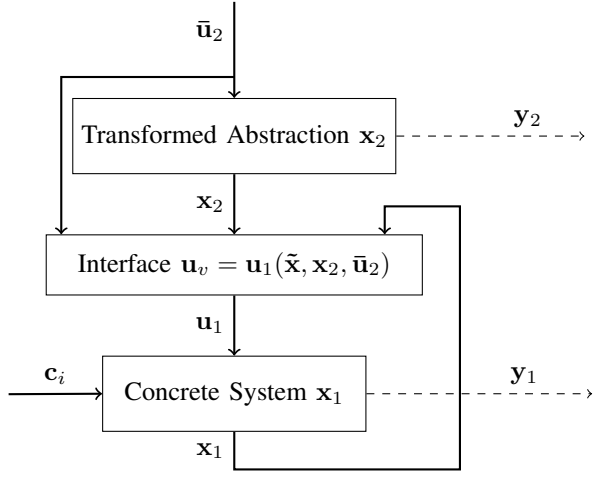
\begin{figure}
    \centering
    \begin{tikzpicture}
        \node[draw, rectangle, minimum width=5cm, minimum height=0.8cm] (In) at (0,0) {Interface \(\mathbf{u}_v=\mathbf{u}_1(\mathbf{\tilde{x}},\mathbf{x}_2,\mathbf{\bar{u}}_2)\)};
        \node[draw, rectangle, below=0.8cm of In, minimum width=3.5cm,  minimum height=1cm] (CS) {Concrete System \(\mathbf{x}_1\)};
        \node[draw, rectangle, above=0.8cm of In, minimum width=3.5cm, minimum height=1cm] (AS) {Transformed Abstraction \(\mathbf{x}_2\)};
    
        \draw[->,thick] ++(0,3.5) to node[pos=0.3,left] {$\mathbf{\bar{u}}_2$} (AS.north) ;
        \draw[->,thick] ++(0,2.5) -| ([xshift=-2.3cm] In.north) ;
        \draw[->,thick] (AS.south) to node[pos=0.5,left] {$\mathbf{x}_2$} (In.north) ;
        \draw[->,thick] (In.south) to node[pos=0.5,left] {$\mathbf{u}_1$}  (CS.north) ;
        \draw[->,thick] (CS.south) to node[pos=0.5,left] {$\mathbf{x}_1$} ++(0,-0.5)-- ++(3,0)-- ++(0,3.5)-- ++(-1,0) --([xshift=2cm] In.north) ;
        \draw[->,thick] ++(-3,-1.72) to node[pos=0.5,above] {$\mathbf{c}_i$} (CS.west) ;
        \draw[->,,dashed] (AS.east) to node[pos=0.7,above] {$\mathbf{y}_2$} ++(2.5,0)  ;
        \draw[->,,dashed] (CS.east) to node[pos=0.7,above] {$\mathbf{y}_1$} ++(3,0)  ;
    \end{tikzpicture}
    \centering
    \caption{Hierarchical control system architecture considered in this work. We extend the robust approximate simulation \cite{Vince2020RASHC} to consider the case where the concrete system is a PWA system.} 
    \label{fig:control_arch}
\end{figure}

%-------------------------------------------

\subsection{Robust Approximate Simulation}\label{subsec:RAS_for_PWA}

\begin{definition} (Approximate Simulation \cite{Girard2006HCAS})\label{def:definition_1}
    A relation \(\mathcal{R}\subseteq\mathbb{R}^m\times \mathbb{R}^n\) is an approximate simulation relation of precision \(\delta\) between \(\Sigma\) and \(\Sigma'\) if for all \((\mathbf{x}_{20},\ \mathbf{x}_{10})\in \mathcal{R}\),
    
\(1)\) \(||\mathbf{y}_{10}-\mathbf{y}_{20}||\leq\delta\), where \(\mathbf{y}_{10}\) and \(\mathbf{y}_{20}\) are the initial values of \(\mathbf{y}_1\) and \(\mathbf{y}_2\) at \(\mathbf{x}_{10}\) and \(\mathbf{x}_{20}\), respectively;
    
    \(2)\) For any state trajectory \(\mathbf{x}_2(t)\) of \(\Sigma'\) such that \(\mathbf{x}_2(0) = \mathbf{x}_{20}\), there exists a state trajectory \(\mathbf{x}_1(t)\) of \(\Sigma\) such that \(\mathbf{x}_1(0) = \mathbf{x}_{10}\) and \((\mathbf{x}_2(t), \mathbf{x}_1(t))\in \mathcal{R}\), for all \(t \geq 0\). 
\end{definition}
\begin{definition} (Robust Approximate Simulation \cite{Vince2021ASTBWBC})\label{def:definition_2}
    The relation \(\mathcal{R}\) is a robust approximate simulation relation if the approximate simulation relation still hold for any disturbances \(\mathbf{c}_i\) in some set \(\mathcal{D}\subseteq \mathbb{R}^n\).
\end{definition}

%-------------------------------------------------------
\section{Main Results}\label{sec:main_results}

In this section, we first analyze the case where the linear system is the abstraction. Then, we extend the established theoretical results to the case where the PWA system is the abstraction. In each case, we first design an \emph{interface} ($\mathbf{u}_v$ in (\ref{eq:entire_system})) that enable the concrete system to track the abstraction, and then design a Lyapunov-like \emph{simulation function} to characterize the formal output tracking errors between the concrete system and the abstraction.

% \subsection{PWA by Linear System under Disturbances}
\subsection{Linear Abstraction for a PWA System Under Disturbances}

In this part, we first design an interface for the joint system (\ref{eq:entire_system}) for the linear abstraction in the form (\ref{eq:linear_abstraction}), and then a Lyapunov-like simulation function is presented to guarantee a formal error bound for the output tracking errors.

\paragraph*{\textbf{Interface}} 
Considering the joint system (\ref{eq:entire_system}) under the linear abstraction (\ref{eq:linear_abstraction}), we propose to design the associated interface as
    \begin{equation}\label{eq:In}
        \mathbf{u}_v(\mathbf{\tilde{x}},\mathbf{x}_2,\mathbf{\bar{u}}_2)=\mathbf{R}_i\mathbf{\bar{u}}_2+(\mathbf{Q}_i+\mathbf{R}_i\mathbf{L})\mathbf{x}_2+\mathbf{K}_i\mathbf{\tilde{x}}
    \end{equation}
for \(i\in\mathcal{I}\), where \(\mathbf{R}_i\in\mathbb{R}^{p\times q}\) is an arbitrary matrix and \(\mathbf{K}_i\in\mathbb{R}^{p\times n}\) is a tunable parameter matrix. Then, note that the joint system (\ref{eq:entire_system}) under the interface (\ref{eq:In}) can be written as
    \begin{equation}\label{eq:entire_under_In}
        \left \{\begin{array}{ll}
        \boldsymbol{\dot{\mathbf{\bar{\omega}}}}=\mathbf{\bar{A}}_i\boldsymbol{\bar{\omega}}+\mathbf{\bar{B}}_{1i}\mathbf{x}_2+\mathbf{\bar{B}}_{2i}\mathbf{\bar{u}}_2+\mathbf{\bar{c}}_i &  \\
        \mathbf{e}=\mathbf{\bar{C}}_i\boldsymbol{\bar{\omega}} &   
        \end{array}\right.
    \end{equation}
where the state \(\boldsymbol{\bar{\omega}}=\begin{bmatrix}
    \boldsymbol{\omega}^T & 1 \\
\end{bmatrix}^T\), \(\mathbf{\bar{A}}_i=\begin{bmatrix}
    \mathbf{A}'_i & \mathbf{0} \\
    \mathbf{0} & 0
\end{bmatrix}\), \(\mathbf{A}'_i=\begin{bmatrix}
    \mathbf{A}_i+\mathbf{B}_i\mathbf{K}_i & \mathbf{0} \\
    \mathbf{0} & \mathbf{F}+\mathbf{G}\mathbf{L}
\end{bmatrix}\), \(\mathbf{\bar{B}}_{1i}=\begin{bmatrix}
    \mathbf{B}'_{1i} \\
    \mathbf{0}
\end{bmatrix}\), \(\mathbf{B}'_{1i}=\begin{bmatrix}
    (\mathbf{B}_i\mathbf{R}-\mathbf{P}_i\mathbf{G})\mathbf{L} \\
    \mathbf{0}
\end{bmatrix}\), \(\mathbf{\bar{B}}_{2i}=\begin{bmatrix}
    \mathbf{B}'_{2i} \\
    \mathbf{0}
\end{bmatrix}\), \(\mathbf{B}'_{2i}=\begin{bmatrix}
    \mathbf{B}_i\mathbf{R}-\mathbf{P}_i\mathbf{G} \\
    \mathbf{G}
\end{bmatrix}\), \(\mathbf{\bar{c}}_i=\begin{bmatrix}
    \mathbf{c'}_i^T & 0 
\end{bmatrix}^T\), \(\mathbf{c}'_i=\begin{bmatrix}
    \mathbf{c}_i \\
    \mathbf{0}
\end{bmatrix}\), \(\mathbf{\bar{C}}_i=\begin{bmatrix}
    \mathbf{C}'_i & 0 
\end{bmatrix}\), \(\mathbf{C}'_i=\begin{bmatrix}
    \mathbf{C}_i & \mathbf{0}
\end{bmatrix}\), and the joint partition is defined as (\ref{eq:entire_system_partition_1}).

\paragraph*{\textbf{Simulation Function}} 
In order to guarantee a formal error bound, we propose a Lyapunov-like simulation function as
    \begin{equation}\label{eq:SF_0}
        \mathcal{V}(\boldsymbol{\omega}):=\frac{1}{\kappa}\sqrt{\boldsymbol{\omega}^T\mathbf{M}_i\boldsymbol{\omega}},\ \mbox{for}\ \boldsymbol{\omega}\in \mathbf{\Omega}_i,\ i\in \mathcal{I}_0
    \end{equation}
    \begin{equation}\label{eq:SF_1}
        \mathcal{V}(\boldsymbol{\bar{\omega}}):=\frac{1}{\kappa}\sqrt{\boldsymbol{\bar{\omega}}^T\mathbf{\bar{M}}_i\boldsymbol{\bar{\omega}}},\ \mbox{for}\ \boldsymbol{\bar{\omega}}\in \mathbf{\bar{\Omega}}_i,\ i\in \mathcal{I}_1
    \end{equation}
where \(\kappa>0\) is some adjustable parameter, each \(\mathbf{\bar{M}}_i=\mathbf{\bar{J}}_i^T\mathbf{T}\mathbf{\bar{J}}_i\) is a diagonal block matrix, i.e., \(\mathbf{\bar{M}}_i=\begin{bmatrix}
    \mathbf{M}_i & \mathbf{0} \\
    \mathbf{0} & m_i
\end{bmatrix}\) with \(\mathbf{M}_i\in\mathbb{R}^{(n+m)\times(n+m)}\) being positive definite and \(m_i>0\), \(\mathbf{\bar{J}}_i\) is the continuity matrix of the joint system (\ref{eq:entire_system}), and \(\mathbf{T}\) is some symmetric free parameter matrix, and \(\mathbf{\bar{M}}_i\) should satisfy the following linear matrix inequalities: 
    \begin{equation}\label{eq:LMI_M}
        \mathbf{\bar{M}}_i-\mathbf{\bar{C}}_i^T\mathbf{\bar{C}}_i\geq 0,\ \mathbf{\bar{M}}_i-\mathbf{\bar{E}}_i^T\mathbf{U}_i\mathbf{\bar{E}}_i>0
    \end{equation}
    \begin{equation}\label{eq:LMI_stable}
        \mathbf{\bar{A}}_i^T\mathbf{\bar{M}}_i+\mathbf{\bar{M}}_i\mathbf{\bar{A}}_i+\mathbf{\bar{E}}_i^T\mathbf{W}_i\mathbf{\bar{E}}_i+\boldsymbol{\bar{\lambda}}\mathbf{\bar{M}}_i\leq 0
    \end{equation}
where \(\boldsymbol{\bar{\lambda}}=\begin{bmatrix}
    \lambda\mathbf{I} & \mathbf{0} \\
    \mathbf{0} & 0
\end{bmatrix}\), \(\lambda>0\) is some parameter free to choose. \(\mathbf{\bar{E}}_i\) is the cell bounding in (\ref{eq:entire_system_partition_1}), and \(\mathbf{U}_i\) and \(\mathbf{W}_i\) are symmetric free parameter matrices with nonnegative entries.\\

\begin{theorem} \label{theorem_1}

\ Assume there exists a matrix \(\mathbf{P}_i\in\mathbb{R}^{n\times m}\) and a matrix \(\mathbf{Q}_i\in\mathbb{R}^{p\times m}\) such that the linear matrix equations in (\ref{approximate_relation_i}) hold for \(i\in \mathcal{I}\). Then, for the associated interface given by (\ref{eq:In}), there exists a robust approximate simulation relation of \(\Sigma\) by \(\Sigma'\), of a precision \(\delta\) as
    \begin{equation}\label{eq:error_bound}
        \delta:= \left \{\begin{array}{ll}
        \ \kappa\mathcal{V}(\boldsymbol{\omega}), &  \ \ \mathcal{V}(\boldsymbol{\omega})\geq b_0,\ i\in\mathcal{I}_0\\\
        \kappa b_0, & \ \ \mathcal{V}(\boldsymbol{\omega})< b_0,\ i\in\mathcal{I}_0\\\
        \kappa\mathcal{V}(\boldsymbol{\bar{\omega}}), &  \ \ \mathcal{V}(\boldsymbol{\bar{\omega}})\geq b_1,\ i\in\mathcal{I}_1\\\
        \kappa b_1, &     \ \ \mathcal{V}(\boldsymbol{\bar{\omega}})< b_1,\ i\in\mathcal{I}_1
        \end{array}\right.
    \end{equation}
where \(b_0=\gamma_1(||\mathbf{\bar{u}}_2||_{\infty})+\gamma_2(||\mathbf{c}'_i||_{\infty})+\gamma_3(||\mathbf{x}_2||_{\infty})\), \(b_1=\gamma_1(||\mathbf{\bar{u}}_2||_{\infty})+\gamma_2(||\mathbf{\bar{c}}_i||_{\infty})+\gamma_3(||\mathbf{x}_2||_{\infty})+\sqrt{m_i}\) with \(\gamma_1(\cdot)\), \(\gamma_2(\cdot)\) and \(\gamma_3(\cdot)\) being some class-$\mathcal{K}$ functions, and \(\mathcal{V}(\boldsymbol{\omega})\) and \(\mathcal{V}(\boldsymbol{\bar{\omega}})\) are Lyapunov-like simulation function as in (\ref{eq:SF_0}) and (\ref{eq:SF_1}), respectively. \\
\end{theorem}
\begin{proof}
First, note that the simulation function (\ref{eq:SF_1}) can bound the output error under (\ref{eq:LMI_M}) by
    \begin{equation}\label{error_bound}
        \mathcal{V}(\boldsymbol{\bar{\omega}})=\frac{1}{\kappa}\sqrt{\boldsymbol{\bar{\omega}}^T\mathbf{\bar{M}}_i\boldsymbol{\bar{\omega}}}\geq \frac{1}{\kappa}||\mathbf{\bar{C}}_i\boldsymbol{\bar{\omega}}||=\frac{1}{\kappa}||\mathbf{e}||
    \end{equation}
Similarly, the simulation function (\ref{eq:SF_0}) can bound the output error by the reduced condition \(\mathbf{M}_i-\mathbf{C}_i'^T\mathbf{C}_i'\geq 0\) in (\ref{eq:LMI_M}). 

For the case \(i\in\mathcal{I}_1\), take the directional derivative of the simulation function (\ref{eq:SF_1}) along (\ref{eq:entire_under_In}),
    \begin{align*}
        &\mathcal{\dot{V}}(\boldsymbol{\bar{\omega}})=\frac{\boldsymbol{\bar{\omega}}^T\mathbf{\bar{M}}_i(\mathbf{\bar{A}}_i\boldsymbol{\bar{\omega}}+\mathbf{\bar{B}}_{1i}\mathbf{x}_2+\mathbf{\bar{B}}_{2i}\mathbf{\bar{u}}_2+\mathbf{\bar{c}}_i)}{\kappa\sqrt{\boldsymbol{\bar{\omega}}^T\mathbf{\bar{M}}_i\boldsymbol{\bar{\omega}}}}\\
        &=\frac{\boldsymbol{\bar{\omega}}^T\mathbf{\bar{M}}_i\mathbf{\bar{A}}_i\boldsymbol{\bar{\omega}}+\boldsymbol{\bar{\omega}}^T\mathbf{\bar{M}}_i\mathbf{\bar{B}}_{1i}\mathbf{x}_2+\boldsymbol{\bar{\omega}}^T\mathbf{\bar{M}}_i\mathbf{\bar{B}}_{2i}\mathbf{\bar{u}}_2+\boldsymbol{\bar{\omega}}^T\mathbf{\bar{M}}_i\mathbf{\bar{c}}_i}{\kappa\sqrt{\boldsymbol{\bar{\omega}}^T\mathbf{\bar{M}}_i\boldsymbol{\bar{\omega}}}}
    \end{align*}
Based on (\ref{eq:LMI_stable}), $\mathbf{\bar{M}}_i\mathbf{\bar{A}}_i \leq -\frac{1}{2}\mathbf{\bar{E}}_i^T\mathbf{W}_i\mathbf{\bar{E}}_i
-\frac{1}{2}\boldsymbol{\bar{\lambda}}\mathbf{\bar{M}}_i$, and thus, 
    \begin{align*}
        \mathcal{\dot{V}}(\boldsymbol{\bar{\omega}})\leq& \frac{\boldsymbol{\bar{\omega}}^T\left[-\frac{1}{2}\mathbf{\bar{E}}_i^T\mathbf{W}_i\mathbf{\bar{E}}_i-\frac{1}{2}\boldsymbol{\bar{\lambda}}\mathbf{\bar{M}}_i\right]\boldsymbol{\bar{\omega}}}{\kappa\sqrt{\boldsymbol{\bar{\omega}}^T\mathbf{\bar{M}}_i\boldsymbol{\bar{\omega}}}}+\frac{1}{\kappa}||\sqrt{\mathbf{\bar{M}}_i}\mathbf{\bar{B}}_{1i}||||\mathbf{x}_2||\\
        &+\frac{1}{\kappa}||\sqrt{\mathbf{\bar{M}}_i}\mathbf{\bar{B}}_{2i}||||\mathbf{\bar{u}}_2||+\frac{1}{\kappa}||\sqrt{\mathbf{\bar{M}}_i}||||\mathbf{\bar{c}}_i||
    \end{align*}
Then, by the definition of the simulation function (\ref{eq:SF_1}) and the parameter matrix \(\boldsymbol{\bar{\lambda}}\), we have
    \begin{align*}
        \mathcal{\dot{V}}(\boldsymbol{\bar{\omega}})\leq& -\frac{\lambda}{2\kappa}\mathcal{V}(\boldsymbol{\bar{\omega}})+\frac{\lambda}{2\kappa^2}\frac{m_i}{\mathcal{V}(\boldsymbol{\bar{\omega}})}+\frac{1}{\kappa}||\sqrt{\mathbf{\bar{M}}_i}\mathbf{\bar{B}}_{1i}||||\mathbf{x}_2||\\
        &+\frac{1}{\kappa}||\sqrt{\mathbf{\bar{M}}_i}\mathbf{\bar{B}}_{2i}||||\mathbf{\bar{u}}_2||+\frac{1}{\kappa}||\sqrt{\mathbf{\bar{M}}_i}||||\mathbf{\bar{c}}_i||
    \end{align*}
According to (\ref{eq:SF_1}), we know that \(\mathcal{V}(\boldsymbol{\bar{\omega}})\geq \frac{1}{\kappa}\sqrt{m_i}\), and therefore we have
    \begin{align*}
        \mathcal{\dot{V}}(\boldsymbol{\bar{\omega}})\leq& -\frac{\lambda}{2\kappa}\mathcal{V}(\boldsymbol{\bar{\omega}})+\frac{\lambda}{2\kappa}\sqrt{m_i}+\frac{1}{\kappa}||\sqrt{\mathbf{\bar{M}}_i}\mathbf{\bar{B}}_{1i}||||\mathbf{x}_2||\\
        &+\frac{1}{\kappa}||\sqrt{\mathbf{\bar{M}}_i}\mathbf{\bar{B}}_{2i}||||\mathbf{\bar{u}}_2||+\frac{1}{\kappa}||\sqrt{\mathbf{\bar{M}}_i}||||\mathbf{\bar{c}}_i||
    \end{align*}
    
For \(\mathcal{\dot{V}}(\boldsymbol{\bar{\omega}})<0\), we need \(\mathcal{V}(\boldsymbol{\bar{\omega}})>\gamma_1(||\mathbf{\bar{u}}_2||_{\infty})+\gamma_2(||\mathbf{\bar{c}}_i||_{\infty})+\gamma_3(||\mathbf{x}_2||_{\infty})+\sqrt{m_i}\), where \(\gamma_1(s)=\frac{2||\sqrt{\mathbf{\bar{M}}_i}\mathbf{\bar{B}}_{2i}||}{\lambda}s\), \(\gamma_2(s)=\frac{2||\sqrt{\mathbf{\bar{M}}_i}||}{\lambda}s\) and \(\gamma_3(s)=\frac{2||\sqrt{\mathbf{\bar{M}}_i}\mathbf{\bar{B}}_{1i}||}{\lambda}s\) are class-$\mathcal{K}$ functions. 

For the case \(i\in\mathcal{I}_0\), the above results will reduce to
    \begin{align*}
        \mathcal{\dot{V}}(\boldsymbol{\omega})\leq& -\frac{\lambda}{2\kappa}\mathcal{V}(\boldsymbol{\omega})+\frac{1}{\kappa}||\sqrt{\mathbf{M}_i}\mathbf{B}'_{1i}||||\mathbf{x}_2||\\
        &+\frac{1}{\kappa}||\sqrt{\mathbf{M}_i}\mathbf{B}'_{2i}||||\mathbf{\bar{u}}_2||+\frac{1}{\kappa}||\sqrt{\mathbf{M}_i}||||\mathbf{c}'_i||
    \end{align*}
    For \(\mathcal{\dot{V}}(\boldsymbol{\omega})<0\), we have \(\mathcal{V}(\boldsymbol{\omega})>\gamma_1(||\mathbf{\bar{u}}_2||_{\infty})+\gamma_2(||\mathbf{c}'_i||_{\infty})+\gamma_3(||\mathbf{x}_2||_{\infty})\), where \(\gamma_1(s)=\frac{2||\sqrt{\mathbf{M}_i}\mathbf{B}'_{2i}||}{\lambda}s\), \(\gamma_2(s)=\frac{2||\sqrt{\mathbf{M}_i}||}{\lambda}s\) and \(\gamma_3(s)=\frac{2||\sqrt{\mathbf{M}_i}\mathbf{B}'_{1i}||}{\lambda}s\) are class-$\mathcal{K}$ functions. 
    
    Therefore, we know that for \(i\in\mathcal{I}_0\), the set 
    \begin{equation}
        \mathcal{R}_0=\{\boldsymbol{\omega}\in\boldsymbol{\Omega}_i|\ \mathcal{V}(\boldsymbol{\omega})\leq b_0\}
    \end{equation}
and for \(i\in\mathcal{I}_1\), the set 
    \begin{equation}
        \mathcal{R}_1=\{\boldsymbol{\bar{\omega}}\in\boldsymbol{\bar{\Omega}}_i|\ \mathcal{V}(\boldsymbol{\bar{\omega}})\leq b_1\}
    \end{equation}
are both forward invariant \cite{khalil2002nonlinear}. Thus, we know that \(\mathcal{R}_0\) and \(\mathcal{R}_1\) satisfy robust approximate simulation relation with precision \(\kappa b_0\) and \(\kappa b_1\), respectively. Furthermore, from (\ref{error_bound}), we also know that \(\frac{1}{\kappa}||\mathbf{y}_1-\mathbf{y}_2||\leq\mathcal{V}(\boldsymbol{\omega})\) for \(i\in\mathcal{I}_0\) and \(\frac{1}{\kappa}||\mathbf{y}_1-\mathbf{y}_2||\leq\mathcal{V}(\boldsymbol{\bar{\omega}})\) for \(i\in\mathcal{I}_1\). Thus, we know that \(\kappa\mathcal{V}(\boldsymbol{\omega})\) and \(\kappa\mathcal{V}(\boldsymbol{\bar{\omega}})\) can bound the output errors and the robust approximate simulation relation can be satisfied when \(\mathcal{V}(\boldsymbol{\omega})>b_0\) and \(\mathcal{V}(\boldsymbol{\bar{\omega}})>b_1\). Thus, all the cases in (\ref{eq:error_bound}) now have been proven and this completes the proof.
\end{proof}

Note that the error bound in \eqref{eq:error_bound} is dependent on $||\mathbf{x}_2||_{\infty}$. However, it can be computed since $\mathbf{x}_2$ is always available.

% \begin{remark}
%     The cell bounding \(\mathbf{\bar{E}}_i=\left[\begin{array}{cc}
%         \mathbf{E}'_i & \mathbf{f}'_i
%     \end{array}\right]\) and the boundary constraint matrix \(\mathbf{\bar{J}}_i=\left[\begin{array}{cc}
%         \mathbf{J}_i & h_i
%     \end{array}\right]\) will reduce to \(\mathbf{E}'_i\) and \(\mathbf{J}_i\), respectively if the partitions contain the origin, i.e. \(\mathbf{f}'_i=\mathbf{0}\) and \(h_i=0\) for \(i\in\mathcal{I}_0\).
% \end{remark}

% \subsection{PWA by PWA under Disturbances}
\subsection{PWA Abstraction for a PWA System Under Disturbances}

In this part, we consider the case where the abstraction is in form \((\ref{eq:PWA_abstraction})\), where we are free to choose the partitions as well as the system matrices on each partition. Compared to the previous part, the main difference in this part lies in the computation of the cell boundings of the joint system since the partitions of the abstract PWA system should be involved in the computation of the joint \(\mathbf{\bar{E}}_{ij}\) as \(\boldsymbol{\bar{\omega}}\in \mathbf{\bar{\Omega}}_{ij}\), where
    \begin{equation}\label{eq:partition_ij}
        \begin{split}
        \mathbf{\bar{\Omega}}_{ij}:&=\Big\{
        \begin{bmatrix}
             \mathbf{E}_i \\
             \mathbf{E}_{cj} 
        \end{bmatrix}(\mathbf{x}_1-\mathbf{P}_i\mathbf{x}_2)+\begin{bmatrix}
             \mathbf{E}_i\mathbf{P}_i \\
             \mathbf{E}_{cj}\mathbf{P}_i
        \end{bmatrix}\mathbf{x}_2\geq\begin{bmatrix}
             \mathbf{f}_i \\
             \mathbf{f}_{cj} 
        \end{bmatrix}\Big\}\\
        &=\Big\{\begin{bmatrix}
            \mathbf{E}_i & \mathbf{E}_i\mathbf{P}_i \\
            \mathbf{E}_{cj}  & \mathbf{E}_{cj}\mathbf{P}_i
        \end{bmatrix}\boldsymbol{\omega}\geq
        \begin{bmatrix}
             \mathbf{f}_i \\
             \mathbf{f}_{cj} 
        \end{bmatrix}\Big\}\\
        &=\{\mathbf{\bar{E}}_{ij}\boldsymbol{\bar{\omega}}\geq\mathbf{0}\}
        \end{split}
    \end{equation}
where \(i\in\mathcal{I}\), \(j\in\mathcal{I}_{a}\) and the matrices \(\mathbf{E}_{cj}\in\mathbb{R}^{l\times n}\) and \(\mathbf{f}_{cj}\in\mathbb{R}^l\) construct the desired partitions of \(\mathbf{x}_2\) in the state space of \(\mathbf{x}_1\) (i.e., \(\mathbf{x}_1\in \mathcal{X}_1^{i}:=\{\mathbf{x}_1\in\mathbb{R}^n|\ \mathbf{E}_i\mathbf{x}_1\geq\mathbf{f}_i\}\),  \(i\in\mathcal{I}\)). In particular, the desired partitions of \(\mathbf{x}_2\) in the state space of \(\mathbf{x}_1\) can be written in closed form as \(\mathbf{E}_{cj}\mathbf{x}_1\geq\mathbf{f}_{cj}\), i.e., \(\mathbf{E}_{cj}\mathbf{\tilde{x}}+\mathbf{E}_{cj}\mathbf{P}_i\mathbf{x}_2\geq\mathbf{f}_{cj}\), for \(j\in\mathcal{I}_a\).

In fact, the partition of \(\mathbf{x}_2\) can be viewed as \(\mathbf{E}_{aj}\mathbf{x}_2\geq\mathbf{f}_{aj}\) as in (\ref{eq:PWA_abstraction}), where \(\mathbf{E}_{aj}=\mathbf{E}_{cj}\mathbf{P}_i\) and \(\mathbf{f}_{aj}=\mathbf{f}_{cj}-\mathbf{E}_{cj}\mathbf{\tilde{x}}\). It is not surprising that the partition of \(\mathbf{x}_2\) varies with \(\mathbf{\tilde{x}}\) because we know that there must exist some linear transformation \(\mathbf{\Pi}\in\mathbb{R}^{m\times n}\) between two different spaces, i.e. \(\mathbf{x}_2=\mathbf{\Pi}\mathbf{x}_1\) from Lemma 1 of \cite{Sadraddini2019SBPTAOC}. However, this relation cannot hold for our case where the disturbances exist in the state space of \(\mathbf{x}_1\) or it may hold with very restrictive conditions (see proposition 4.1 of \cite{RAKOVIC2007SIORC}). The relations \(\mathbf{E}_{aj}=\mathbf{E}_{cj}\mathbf{P}_i\) and \(\mathbf{f}_{aj}=\mathbf{f}_{cj}-\mathbf{E}_{cj}\mathbf{\tilde{x}}\) can be viewed as an online estimate to reshape the partitions of \(\mathbf{x}_2\) in our case. 

\begin{remark}
    It seems unreasonable to express the partitions of \(\mathbf{x}_2\) in the state space of \(\mathbf{x}_1\) as \(\mathbf{E}_{cj}\mathbf{x}_1\geq\mathbf{f}_{cj}\) for \(j\in\mathcal{I}_a\) because the combinations of the successive partitions of \(\mathbf{x}_1\) may not always be convex and thus the partitions may not be written in this closed form. However, we can always convexify the combinations of the successive partitions of \(\mathbf{x}_1\) by manually adding more edges to the original partitions. In this case, we can describe the joint partitions in terms of the equivalent partitions of \(\mathbf{x}_1\). 
\end{remark}

\begin{remark}
    In (\ref{eq:partition_ij}), we use the joint partitions of the joint system because the partitions of \(\mathbf{x}_2\) are hard to determine. However, we know that the partitions of \(\mathbf{x}_1\) can be predefined and thus the desired partitions of \(\mathbf{x}_2\) in the state space of \(\mathbf{x}_1\) can also be determined. In order to do so, we first use the conditions in (\ref{approximate_relation_i}) to check the approximate simulation relations between the concrete system (\ref{eq:PWA_concrete}) and the abstraction (\ref{eq:PWA_abstraction}). Once the relations can be found, the corresponding successive partitions of \(\mathbf{x}_1\) can be combined together. Then, we can analyze these combinations in two cases. For the case that the combinations of the successive partitions of \(\mathbf{x}_1\) are still convex, the combined partitions of \(\mathbf{x}_1\) can be used to describe the desired partitions of \(\mathbf{x}_2\), i.e., \(\mathbf{E}_{cj}\mathbf{\tilde{x}}+\mathbf{E}_{cj}\mathbf{P}_i\mathbf{x}_2\geq\mathbf{f}_{cj}\), for \(j\in\mathcal{I}_a\). For the case that the combinations of the successive partitions of \(\mathbf{x}_1\) are non-convex, we can first convexify the combinations of the successive partitions by adding edges and then follow the same processes as the previous case. Designing systematic convexification approaches is a subject of future research. 
\end{remark}

\paragraph*{\textbf{Interface}} 
Considering the abstraction being in the form as (\ref{eq:PWA_abstraction}), an associated interface can be designed as
\begin{equation}\label{eq:In_ij}
    \mathbf{u}_v(\mathbf{\tilde{x}},\mathbf{x}_2,\mathbf{\bar{u}}_2)=\mathbf{R}_{ij}\mathbf{\bar{u}}_2+(\mathbf{Q}_i+\mathbf{R}_{ij}\mathbf{L}_j)\mathbf{x}_2+\mathbf{K}_i\mathbf{\tilde{x}}
\end{equation}
for the pairs \((i,j)\), \(i\in\mathcal{I}\) and \(j\in \mathcal{I}_a\), where \(\mathbf{R}_{ij}\in\mathbb{R}^{p\times q}\) is an arbitrary matrix and \(\mathbf{K}_{i}\in\mathbb{R}^{p\times n}\) is a tunable parameter matrix. Then, joint system under the interface (\ref{eq:In_ij}) can be written in the following form
\begin{equation}\label{eq:entire_under_In_ij}
        \left \{\begin{array}{ll}
        \boldsymbol{\dot{\mathbf{\bar{\omega}}}}=\mathbf{\bar{A}}_{ij}\boldsymbol{\bar{\omega}}+\mathbf{\bar{B}}_{1ij}\mathbf{x}_2+\mathbf{\bar{B}}_{2ij}\mathbf{\bar{u}}_2+\mathbf{\bar{c}}_i &  \\
        \mathbf{e}=\mathbf{\bar{C}}_{ij}\boldsymbol{\bar{\omega}} &    
        \end{array}\right.
    \end{equation}
for \(i\in \mathcal{I}\), \(j\in\mathcal{I}_a\), where the state \(\boldsymbol{\bar{\omega}}=\begin{bmatrix}
    \boldsymbol{\omega}^T & 1 \\
\end{bmatrix}^T\), \(\mathbf{\bar{A}}_{ij}=\begin{bmatrix}
    \mathbf{A}'_{ij} & \mathbf{0} \\
    \mathbf{0} & 0
\end{bmatrix}\), \(\mathbf{A}'_{ij}=\begin{bmatrix}
    \mathbf{A}_i+\mathbf{B}_i\mathbf{K}_i & \mathbf{0} \\
    \mathbf{0} & \mathbf{F}_j+\mathbf{G}_j\mathbf{L}_j
\end{bmatrix}\), \(\mathbf{\bar{B}}_{1ij}=\begin{bmatrix}
    \mathbf{B}'_{1ij} \\
    \mathbf{0}
\end{bmatrix}\), \(\mathbf{B}'_{1ij}=\begin{bmatrix}
    (\mathbf{B}_i\mathbf{R}_{ij}-\mathbf{P}_i\mathbf{G}_j)\mathbf{L}_j \\
    \mathbf{0}
\end{bmatrix}\), \(\mathbf{\bar{B}}_{2ij}=\begin{bmatrix}
    \mathbf{B}'_{2ij} \\
    \mathbf{0}
\end{bmatrix}\), \(\mathbf{B}'_{2ij}=\begin{bmatrix}
    \mathbf{B}_i\mathbf{R}_{ij}-\mathbf{P}_i\mathbf{G}_j \\
    \mathbf{G}_j
\end{bmatrix}\), \(\mathbf{\bar{c}}_i=\begin{bmatrix}
    \mathbf{c'}_i^T & 0 
\end{bmatrix}^T\), \(\mathbf{c}'_{i}=\begin{bmatrix}
    \mathbf{c}_{i} \\
    \mathbf{0}
\end{bmatrix}\), \(\mathbf{\bar{C}}_{ij}=\begin{bmatrix}
    \mathbf{C}'_{ij} & 0 
\end{bmatrix}\), \(\mathbf{C}'_{ij}=\begin{bmatrix}
    \mathbf{C}_{ij} & \mathbf{0}
\end{bmatrix}\) and the joint partition is defined as (\ref{eq:partition_ij}). 

\paragraph*{\textbf{Simulation Function}} 
Similar as the linear abstraction case in previous part, to guarantee a formal error bound, the Lyapunov-like
simulation function is presented as
\begin{equation}\label{eq:SF_ij_0}
    \mathcal{V}(\boldsymbol{\omega}):=\frac{1}{\kappa}\sqrt{\boldsymbol{\omega}^T\mathbf{M}_{ij}\boldsymbol{\omega}},\ \mbox{for}\ \boldsymbol{\omega}\in \mathbf{\Omega}_{ij},\ i\in \mathcal{I}_0, j\in \mathcal{I}_{a0}
\end{equation}
\begin{equation}\label{eq:SF_ij_1}
    \mathcal{V}(\boldsymbol{\bar{\omega}}):=\frac{1}{\kappa}\sqrt{\boldsymbol{\bar{\omega}}^T\mathbf{\bar{M}}_{ij}\boldsymbol{\bar{\omega}}},\ \mbox{for}\ \boldsymbol{\bar{\omega}}\in \mathbf{\bar{\Omega}}_{ij},\ i\in \mathcal{I}_1, j\in \mathcal{I}_{a1}
\end{equation}
where \(\kappa>0\) is some adjustable parameter, with each \(\mathbf{\bar{M}}_{ij}=\mathbf{\bar{J}}_{ij}^T\mathbf{T}\mathbf{\bar{J}}_{ij}=\begin{bmatrix}
    \mathbf{M}_{ij} & \mathbf{0} \\
    \mathbf{0} & m_{ij}
\end{bmatrix}\), where \(\mathbf{M}_{ij}\) is a positive definite matrix, \(m_{ij}>0\) and the matrix \(\mathbf{\bar{M}}_{ij}\) should satisfy the following linear matrix inequalities: 
\begin{equation}\label{eq:LMI_Mij}
    \mathbf{\bar{M}}_{ij}-\mathbf{\bar{C}}^T_{ij}\mathbf{\bar{C}}_{ij}\geq 0,\ \mathbf{\bar{M}}_{ij}-\mathbf{\bar{E}}_{ij}^T\mathbf{U}_{ij}\mathbf{\bar{E}}_{ij}>0
\end{equation}
\begin{equation}\label{eq:LMI_stable_ij}
        \mathbf{\bar{A}}_{ij}^T\mathbf{\bar{M}}_{ij}+\mathbf{\bar{M}}_{ij}\mathbf{\bar{A}}_{ij}+\mathbf{\bar{E}}_{ij}^T\mathbf{W}_{ij}\mathbf{\bar{E}}_{ij}+\boldsymbol{\bar{\lambda}}\mathbf{\bar{M}}_{ij}\leq 0
\end{equation}
for the pairs \((i,j)\), \(i\in \mathcal{I}\) and \(j\in \mathcal{I}_a\), where \(\boldsymbol{\bar{\lambda}}=\begin{bmatrix}
    \lambda\mathbf{I} & \mathbf{0} \\
    \mathbf{0} & 0
\end{bmatrix}\), for \(\lambda>0\) free to be selected, \(\mathbf{\bar{E}}_{ij}\) is the cell bounding in (\ref{eq:partition_ij}), \(\mathbf{U}_{ij}\) and \(\mathbf{W}_{ij}\) are some symmetric free parameter matrices with nonnegative entries.

With the joint partitions in (\ref{eq:partition_ij}),
we now provide a generalized version of Theorem 1 in the following theorem.\\

\begin{theorem} \label{theorem_2}
    \ For a given tuple of the concrete system  \((\mathbf{A}_i,\ \mathbf{B}_i,\ \mathbf{C}_i)\), \(i\in \mathcal{I}\), assume there exists a matrix \(\mathbf{P}_i\in\mathbb{R}^{n\times m}\) and a matrix \(\mathbf{Q}_i\in\mathbb{R}^{p\times m}\) such that the linear matrix equations in (\ref{approximate_relation_ij}) hold for a pair of \((\mathbf{F}_j,\ \mathbf{H}_j)\) of the abstraction, \(j\in\mathcal{I}_a\). Then, for the associated interface given by (\ref{eq:In_ij}), there exists a robust approximate simulation relation of \(\Sigma\) by \(\Sigma''\), of the precision \(\delta\) as
    \begin{equation}\label{eq:error_bound_ij}
        \delta:= \left \{\begin{array}{ll}
        \ \kappa\mathcal{V}(\boldsymbol{\omega}), &  \ \ \mathcal{V}(\boldsymbol{\omega})\geq b_0,\  i\in\mathcal{I}_0,\ j\in \mathcal{I}_{a0}\\\
        \kappa b_0, & \ \ \mathcal{V}(\boldsymbol{\omega})< b_0,\ i\in\mathcal{I}_0,\ j\in \mathcal{I}_{a0}\\\
        \kappa\mathcal{V}(\boldsymbol{\bar{\omega}}), &  \ \ \mathcal{V}(\boldsymbol{\bar{\omega}})\geq b_1,\ i\in\mathcal{I}_1,\ j\in \mathcal{I}_{a1}\\\
        \kappa b_1, &     \ \ \mathcal{V}(\boldsymbol{\bar{\omega}})< b_1,\ i\in\mathcal{I}_1,\ j\in \mathcal{I}_{a1}
        \end{array}\right.
    \end{equation}
where \(b_0=\gamma_1(||\mathbf{\bar{u}}_2||_{\infty})+\gamma_2(||\mathbf{c}'_{i}||_{\infty})+\gamma_3(||\mathbf{x}_2||_{\infty})\) and \(b_1=\gamma_1(||\mathbf{\bar{u}}_2||_{\infty})+\gamma_2(||\mathbf{\bar{c}}_{i}||_{\infty})+\gamma_3(||\mathbf{x}_2||_{\infty})+\sqrt{m_{ij}}\), with \(\gamma_1(\cdot)\), \(\gamma_2(\cdot)\) and \(\gamma_3(\cdot)\) being some class-$\mathcal{K}$ functions, and \(\mathcal{V}(\boldsymbol{\omega})\) and \(\mathcal{V}(\boldsymbol{\bar{\omega}})\) are Lyapunov-like simulation functions as in (\ref{eq:SF_ij_0}) and (\ref{eq:SF_ij_1}), respectively.
\end{theorem}
\begin{proof}
    The proof of Theorem 2 is similar to that of Theorem 1, and thus is omitted here.
\end{proof}

\section{Simulation Examples}\label{sec:simulation}

In this section, we use two simulation examples to illustrate the effectiveness of the proposed robust approximate simulation based control method developed for PWA systems.

%-------------------------------------------

\subsection{Case 1: PWA by Linear}\label{subsec:case_1}

In this part, we consider an example of a robot tracking a 3-part road section. The robot must navigate the path with different dynamics in each part. The output of both concrete and abstract systems represent the position of the robot in the plane. The control objective is to make the robot track the planned path in \(\mathbf{x}_2\), meanwhile the output errors maintain bounded. 

The concrete system was selected as a triple integrator, but its dynamics is different on each part of the path in terms of the input mapping matrix \(\mathbf{B}_i\), i.e.,
\begin{equation}\label{eq:PWA_concrete_simulation}
    \Sigma:\left \{\begin{array}{ll}
    \mathbf{\dot{x}}_1=\mathbf{A}_i\mathbf{x}_1+\mathbf{B}_i\mathbf{u}_1+\mathbf{c}_i &  \\
    \mathbf{y}_1=\mathbf{C}_i\mathbf{x}_1 &    
    \end{array}\right.
\end{equation}
with the parameters of the system as follows 
\[ \mathbf{A}_i=\begin{bmatrix}
    \mathbf{0}_2 & \mathbf{I}_2 & \mathbf{0}_2 \\
    \mathbf{0}_2 & \mathbf{0}_2 & \mathbf{I}_2 \\
    \mathbf{0}_2 & \mathbf{0}_2 & \mathbf{0}_2
\end{bmatrix},\ \mathbf{B}_1=0.5\mathbf{B}_2=2\mathbf{B}_3=\begin{bmatrix}
    \mathbf{0}_2  \\
    \mathbf{0}_2  \\
    \mathbf{I}_2 
\end{bmatrix},\  \]
\[ \mathbf{C}_i=\begin{bmatrix}
    \mathbf{I}_2 & \mathbf{0}_2 & \mathbf{0}_2 
\end{bmatrix},\ \mathbf{c}_i=(-0.1+0.05sin(t))\mathbf{1}_{6\times 1}\] 
for \(i=1,2,3\), and the partitions are set based on the position of the robot as 
\[ \mathbf{E}'_i=\begin{bmatrix}
    \mathbf{E}'_{i1} & \mathbf{0}_2 & \mathbf{0}_2 \\
    \mathbf{0}_2 & \mathbf{0}_2 & \mathbf{0}_2\\
    \mathbf{0}_2 & \mathbf{0}_2 & \mathbf{0}_2
\end{bmatrix},\ \mathbf{f}'_i=\mathbf{0}
, \mbox{ for  \(i=1,2,3\),}\]
 where
\(\mathbf{E}'_{11}=-\mathbf{E}'_{31}=\begin{bmatrix}
    -1 & 1 \\
    -1 & -1
\end{bmatrix}\) and \(\mathbf{E}'_{21}=\begin{bmatrix}
    -1 & 1 \\
    1 & 1
\end{bmatrix}\). 

The abstract system was selected as a single integrator, i.e.,
\begin{equation}\label{eq:PWA_abstract_simulation}
    \Sigma':\left \{\begin{array}{ll}
    \mathbf{\dot{x}}_2=\mathbf{u}_2 &  \\
    \mathbf{y}_2=\mathbf{x}_2 &    
    \end{array}\right.
\end{equation}
where the input transformation law was designed as \(\mathbf{u}_2=-\mathbf{x}_2+\mathbf{\bar{u}}_2\). The control parameters for \(i=1,2,3\) were selected as 
\[ \mathbf{P}_i=\begin{bmatrix}
    \mathbf{I}_2 & \mathbf{0}_2 & \mathbf{0}_2 
\end{bmatrix}^T,\ \mathbf{Q}_i=\mathbf{0}_2,\ \mathbf{R}_i=\mathbf{B}_i^T[\mathbf{B}_i\mathbf{B}_i^T]^{-1}\mathbf{P}_i\mathbf{G},\ \]
\[\mathbf{K}_1=2\mathbf{K}_2=0.5\mathbf{K}_3=-\begin{bmatrix}
    52\mathbf{I}_2 & 52.3\mathbf{I}_2 & 13\mathbf{I}_2 
\end{bmatrix}\]
The observed simulation results under these system and control parameters are shown in Figures 2, 3 and 4.
\begin{figure}[!t]
    \centering
    \includegraphics[width=0.8\linewidth]{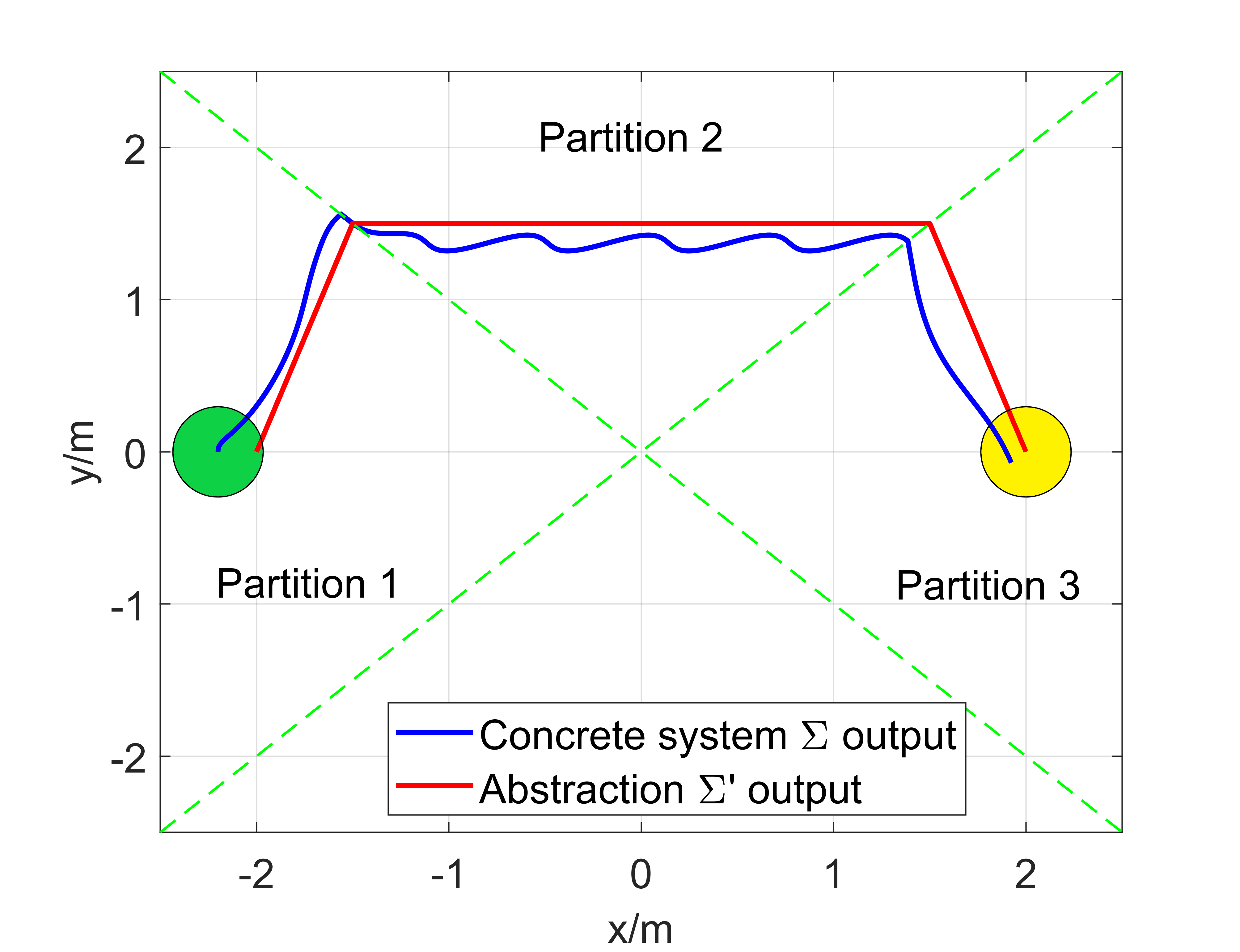}
    \caption{The output of the concrete system \(\Sigma\) and the abstraction \(\Sigma'\), starting from the green circular region and ending at the yellow circular target region. The boundaries of the partitions are shown in green dashed lines.}
    \label{fig:plots_output_tracking1}
\end{figure}
\begin{figure}[!t]
    \centering
    \includegraphics[width=0.8\linewidth]{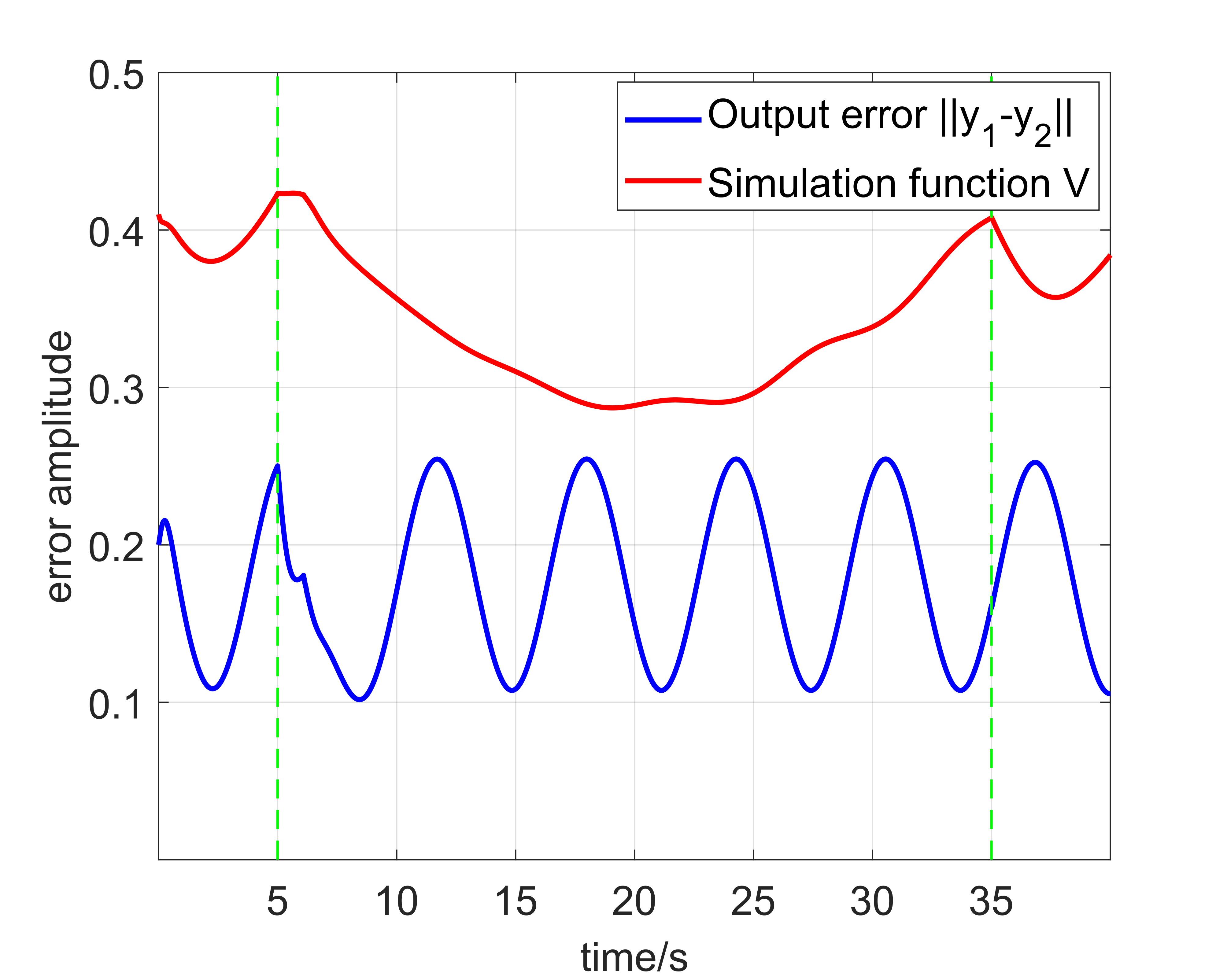}
    \caption{The simulated output error \(||\mathbf{y}_1-\mathbf{y}_2||\) between the concrete system \(\Sigma\) and the abstraction \(\Sigma'\) and the value of simulation function using (\ref{eq:SF_0}) with $\kappa=8$. Boundary crossing soon after green dashed lines.}
    \label{fig:plots_simulation_function1}
\end{figure}
\begin{figure}[!t]
    \centering
    \includegraphics[width=0.8\linewidth]{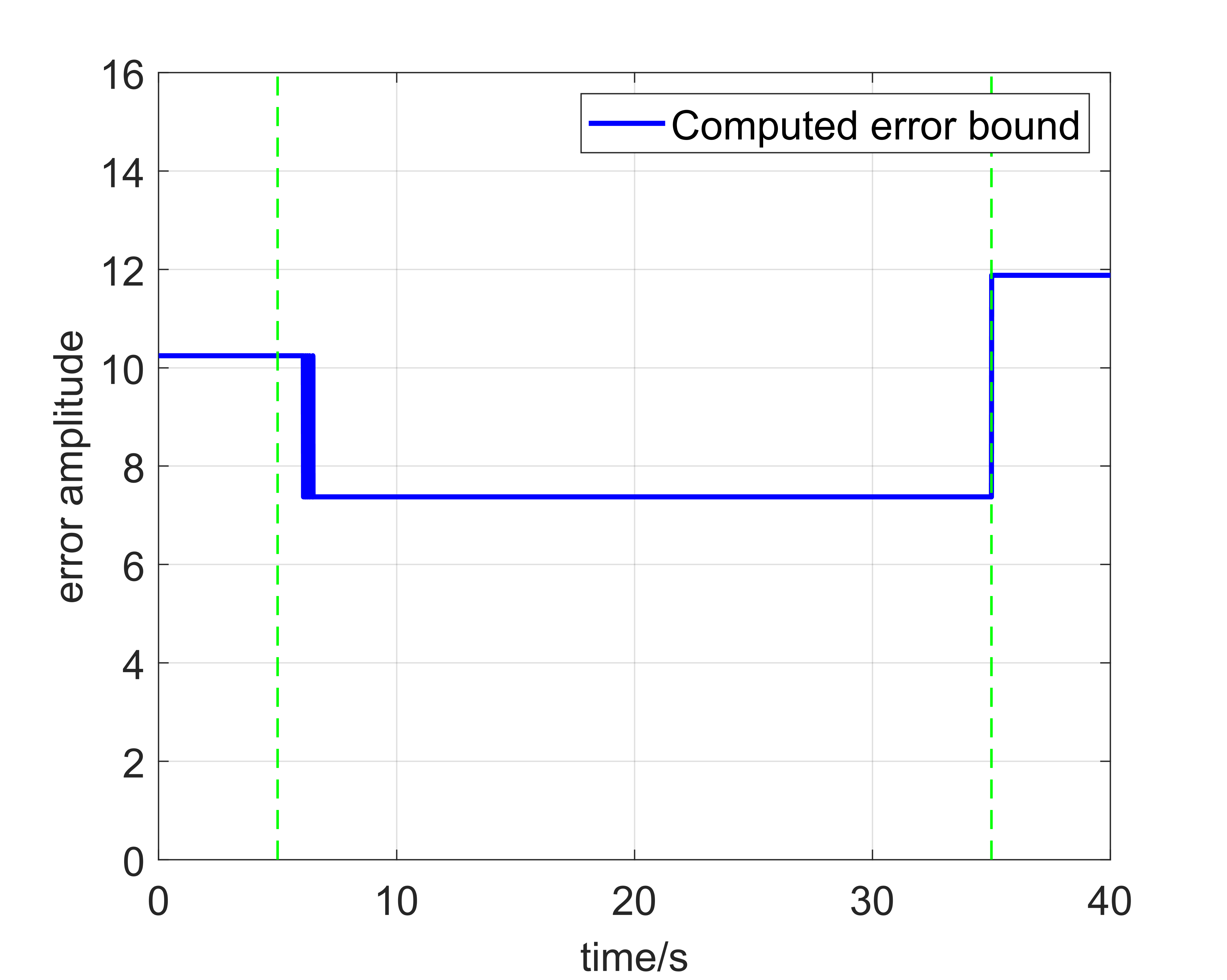}
    \caption{The computed error bound using \(b_0\). Boundary crossing soon after green dashed lines.}
    \label{fig:plots_error_bounds1}
\end{figure}

In Figure 2, it can be seen that the output of the concrete system \(\Sigma\) starts from the green colored initial region and successfully tracks the output of the abstraction \(\Sigma'\) and the two trajectories remain considerably close to each other until the robot reaches the yellow colored target region. More precisely, the output tracking error is shown in Figure 3, which is bounded by the simulation function \(\mathcal{V}\) (obtained with \(\kappa=8\)). Thus, the proposed method is effective. In Figure 3, it can also be seen that when \(\Sigma'\) crosses the boundary of the partitions, the value of the simulation function increases. The reason is that \(\mathbf{x}_2\) is not slowly varying and the amplitude of \(||\mathbf{x}_2||\) also increases at the boundary, which will cause the joint state \(\boldsymbol{\omega}\) to increase. In Figure 4, the computed error bound is relatively large with respect to the simulation function and the value of this computed bound also changes around the instants that \(\Sigma'\) and \(\Sigma\) cross a boundary. At the time instant around 7 sec, there is a chattered switching between two different values of the error bounds, which is because the concrete system temporarily moves back and forth across a boundary of two different partitions. From Figure 3 and 4, we can also see that the simulation function remains bounded within the computed error bound.

\subsection{Case 2: PWA by PWA}\label{subsec:case_2}

In this part, we consider the case that a simpler PWA system is the abstraction for the concrete PWA system. Therefore, we consider a more complicated tracking example with 5-parts road section, while the control objective is similar to the Case 1 discussed before. 

The concrete PWA system was chosen to be in the same form as (\ref{eq:PWA_concrete_simulation}), but its parameters were selected as
\[ \mathbf{A}_i=\eta_{i1}\begin{bmatrix}
    \mathbf{I}_2 & \mathbf{I}_2\\
    \mathbf{0}_2 & \mathbf{I}_2
\end{bmatrix},\ \mathbf{B}_i=\eta_{i2}\begin{bmatrix}
    \mathbf{0}_2  \\
    \mathbf{I}_2 
\end{bmatrix},\  \]
\[ \mathbf{C}_i=\begin{bmatrix}
    \mathbf{I}_2 & \mathbf{0}_2
\end{bmatrix},\ \mathbf{c}_i=(-0.1+0.05sin(t))\mathbf{1}_{4\times 1} \]
for \(i=1,2,...,5\). 

The abstraction was selected as a simpler PWA system,
\begin{equation}\label{eq:PWA_abstract_simulation_1}
    \Sigma'':\left \{\begin{array}{ll}
    \mathbf{\dot{x}}_2=\mathbf{F}_j\mathbf{x}_2+\mathbf{G}_j\mathbf{u}_2 &  \\
    \mathbf{y}_2=\mathbf{H}_j\mathbf{x}_2 &    
    \end{array}\right.
\end{equation}
where the parameters were chosen as
\[ \mathbf{F}_j=\eta_{i1}\mathbf{I}_2,\ \mathbf{G}_j=\mathbf{I}_2,\ \mathbf{H}_j=\mathbf{I}_2,\ \]
for \(j=1,2,3\), with \(\eta_{11}=\eta_{21}=1\), \(\eta_{31}=2\), \(\eta_{41}=\eta_{51}=0.5\), \(\eta_{12}=\eta_{22}=1\), \(\eta_{32}=2\), \(\eta_{42}=\eta_{52}=0.5\), and the input transformation law was designed as  \(\mathbf{u}_2=-k_i\mathbf{x}_2+\mathbf{\bar{u}}_2\), where \(k_1=k_2=3\), \(k_3=4\), \(k_4=k_5=2.5\). 

The joint partitions were selected as
\[\mathbf{\bar{E}}_{ij}=\begin{bmatrix}
    \mathbf{E}'_{ij} & \mathbf{0}_2 & \mathbf{f}'_{ij} \\
    \mathbf{0}_2 & \mathbf{0}_2 & \mathbf{0}_{2\times 1}
\end{bmatrix},\ \] 
for \(i=1,...,5\) and \(j=1,2,3\), where \(\mathbf{E}'_{11}=\begin{bmatrix}
    -1 & 0 \\
    0 & 0
\end{bmatrix}\), \(\mathbf{E}'_{21}=\begin{bmatrix}
    1 & 0 \\
    -1 & 0
\end{bmatrix}\), \(\mathbf{E}'_{32}=\begin{bmatrix}
    1 & 0 \\
    -1 & 0
\end{bmatrix}\), \(\mathbf{E}'_{43}=\begin{bmatrix}
    1 & 0 \\
    -1 & 0
\end{bmatrix}\), \(\mathbf{E}'_{53}=\begin{bmatrix}
    1 & 0 \\
    0 & 0
\end{bmatrix}\), \(\mathbf{f}'_{11}=\begin{bmatrix}
     1.5  \\
     0
\end{bmatrix}\), \(\mathbf{f}'_{21}=\begin{bmatrix}
     -1.5  \\
     0.5
\end{bmatrix}\), \(\mathbf{f}'_{32}=\begin{bmatrix}
     -0.5  \\
     -0.5
\end{bmatrix}\), \(\mathbf{f}'_{43}=\begin{bmatrix}
     0.5  \\
     -1.5
\end{bmatrix}\) and \(\mathbf{f}'_{53}=\begin{bmatrix}
     1.5  \\
     0
\end{bmatrix}\). The output of both systems represents the position of the robot in the plane. The control parameters for \(i=1,...,5\) and \(j=1,2,3\) were selected as 
\[ \mathbf{P}_i=\begin{bmatrix}
    \mathbf{I}_2 & \mathbf{0}_2
\end{bmatrix}^T,\ \mathbf{Q}_i=\mathbf{0}_2,\ \mathbf{R}_{ij}=\mathbf{B}^T_i[\mathbf{B}_i\mathbf{B}_i^T]^{-1}\mathbf{P}_i\mathbf{G}_j\]
\[\mathbf{K}_1=\mathbf{K}_2=2\mathbf{K}_3=0.5\mathbf{K}_4=0.5\mathbf{K}_5=-\begin{bmatrix}
    50\mathbf{I}_2 & 10\mathbf{I}_2
\end{bmatrix}\]
Under these system and the control parameters, the observed simulation results are shown in Figures 5, 6 and 7. 

In Figure 5, it can be seen that the output of the concrete system \(\Sigma\) can track the output of the abstraction \(\Sigma''\), starting from the green colored initial region and ending at the yellow colored target region, and in-between, the two trajectories remain close. The output tracking error is shown in Figure 6, which is bounded by the simulation function \(\mathcal{V}\) (obtained with \(\kappa=12\)), which implies that the proposed method is effective. In Figure 6, it can also be seen that the value of simulation function changes when \(\Sigma''\) crosses the boundary of a partition and the reason is similar to that of Case 1. According to Figure 6 and 7, the simulation function is bounded by the computed error bound and the value of this computed error bound also change around the instants where \(\Sigma''\) and \(\Sigma\) cross a boundary.

\begin{figure}[!t]
    \centering
    \includegraphics[width=0.8\linewidth]{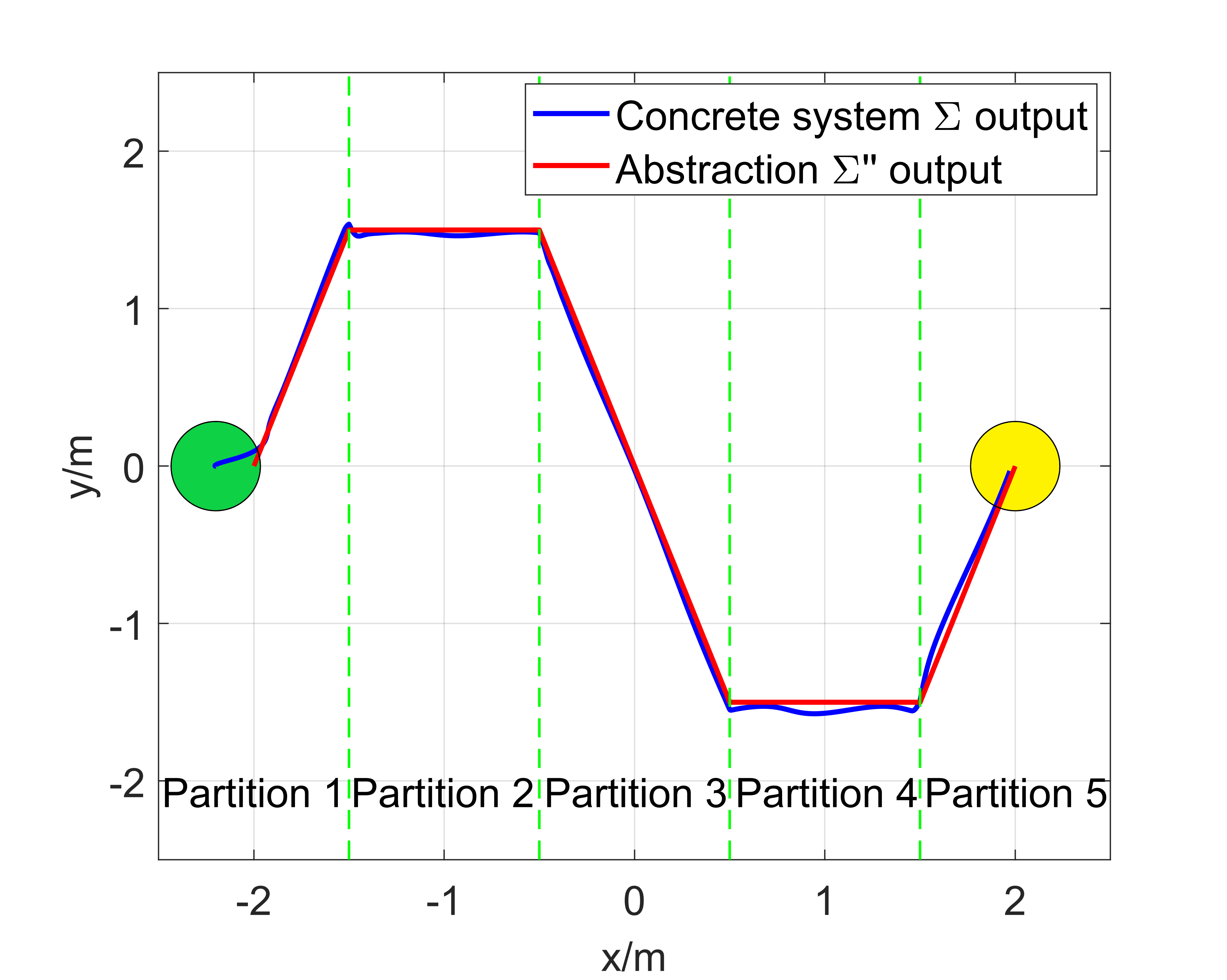}
    \caption{The output of the concrete system \(\Sigma\) and the abstraction \(\Sigma''\), starting from the green circular region and ending at the yellow target region. The boundaries of the partitions are shown in green dashed lines.}
    \label{fig:plots_output_tracking2}
\end{figure}

\begin{figure}[!t]
    \centering
    \includegraphics[width=0.8\linewidth]{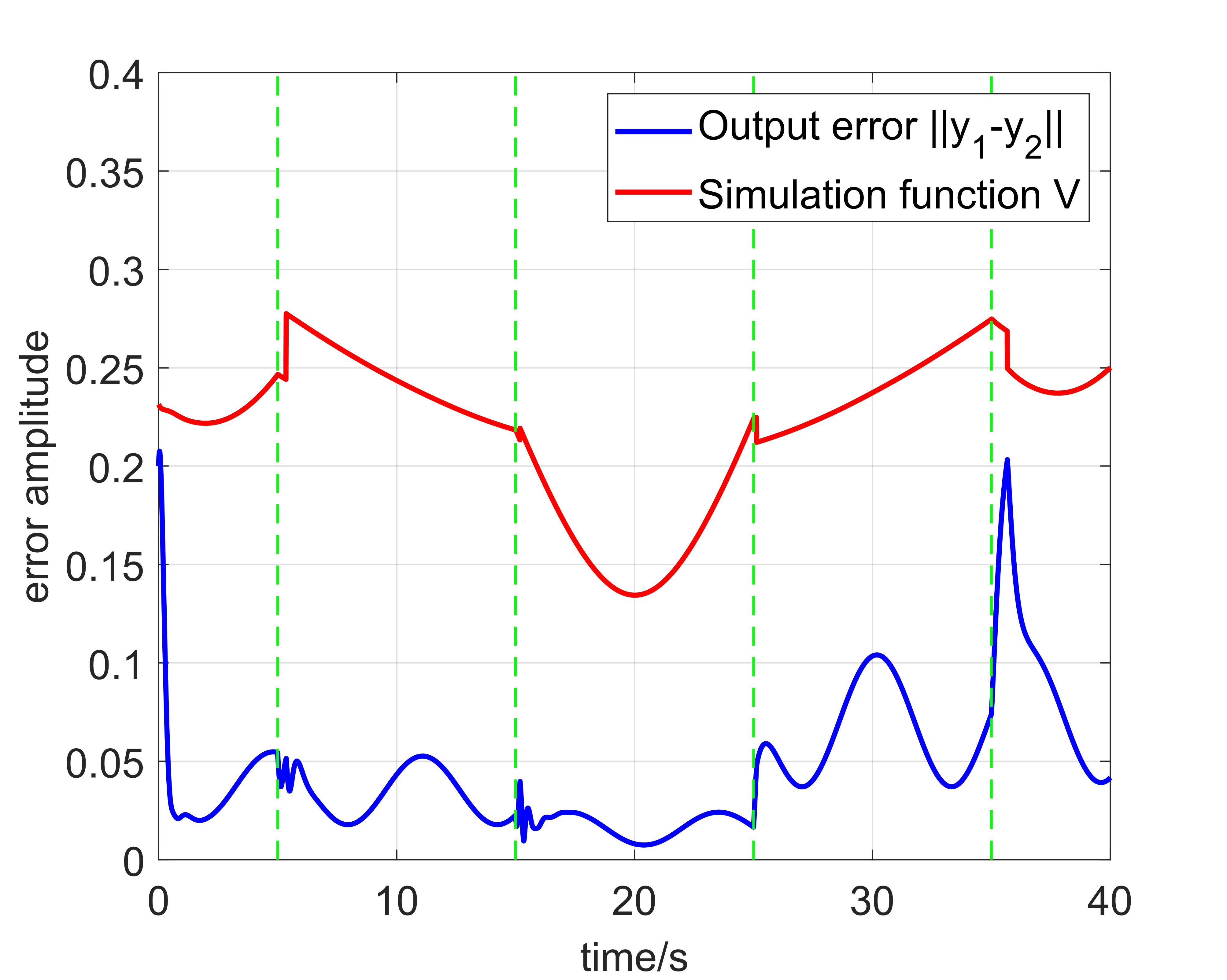}
    \caption{The simulated output error \(||\mathbf{y}_1-\mathbf{y}_2||\) between the concrete system \(\Sigma\) and the abstraction \(\Sigma''\) and the value of simulation function using (\ref{eq:SF_ij_1}) with $\kappa=12$. Boundary crossing soon after green dashed lines.}
    \label{fig:plots_simulation_function2}
\end{figure}

\begin{figure}[!t]
    \centering
    \includegraphics[width=0.8\linewidth]{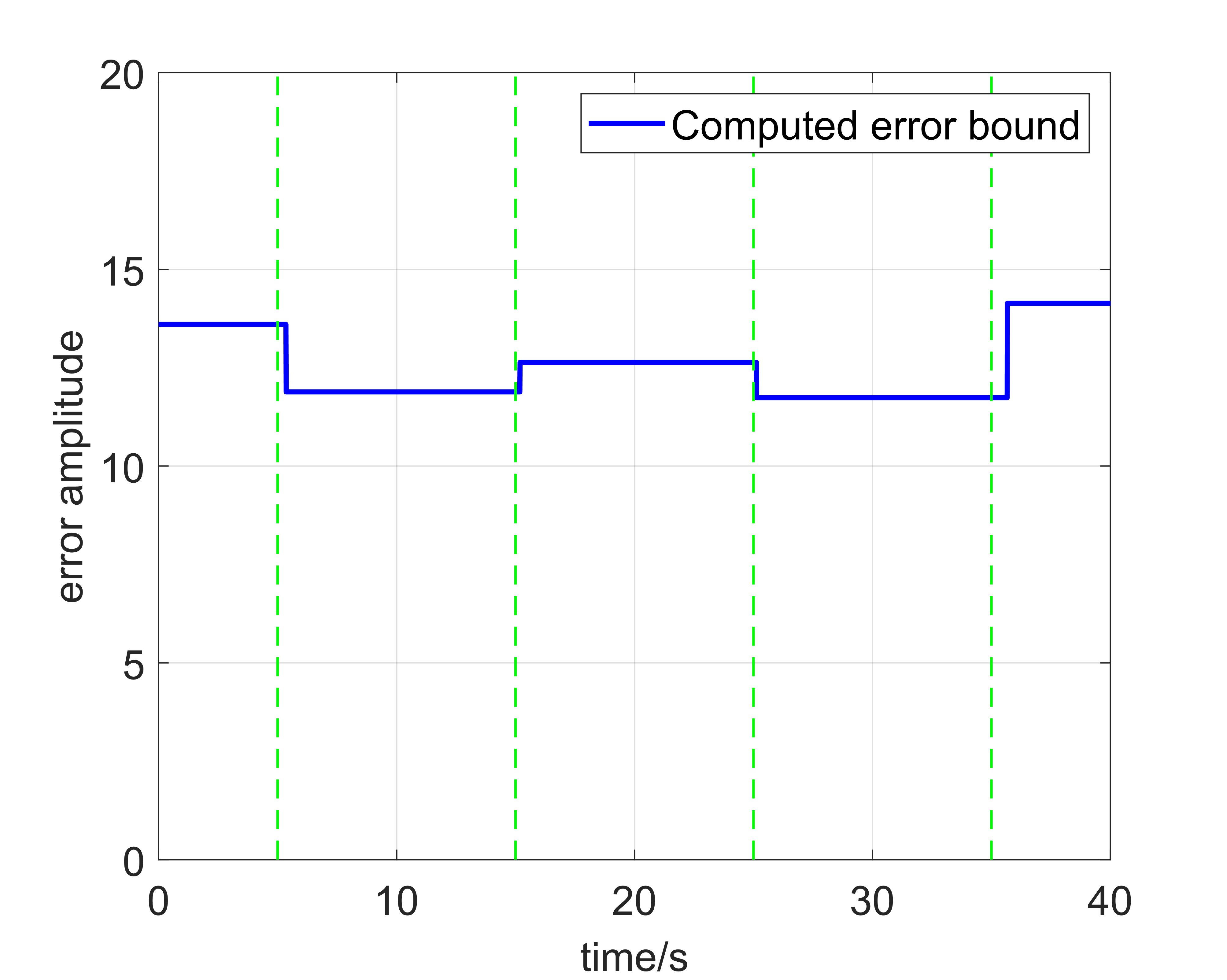}
    \caption{The computed error bound using \(b_1\). Boundary crossing soon after green dashed lines.}
    \label{fig:plots_error_bounds2}
\end{figure}

\section{Conclusion}\label{sec:conclusion}

In this paper, a novel control strategy was presented for the control of PWA systems based on robust approximate simulation framework. First, we designed the interface and the simulation function for a configuration where the concrete system is a known PWA system and the abstraction is a linear system with the system matrices free to choose. Then, the proposed design procedure was generalized for a configuration where the abstraction is a PWA system with the system matrices free to choose under some constraints on its partitions. Finally, we used two simulation examples to  illustrate the effectiveness of the proposed method. Future work aim to improve the established formal error bound such that the proposed approximate simulation based control method can be applied to address more general planning tasks. 

\bibliographystyle{IEEEtran}
\bibliography{references}

\end{document}